\pdfoutput=1 
\documentclass[11pt, reqno]{amsart}
\usepackage{amsmath,amsfonts,amssymb,a4,color,graphics}
\usepackage[latin1]{inputenc}
\usepackage{verbatim}
\usepackage{amsmath,amsfonts,amsthm,amssymb,amsxtra}
\usepackage{amsxtra, amssymb, mathrsfs}

\newcommand{\B}{\mathscr{B}}
\newcommand{\eps}{\varepsilon}
\newcommand{\U}{\mathcal{U}}

\def\ga {\alpha}

\def\gl {\lambda}

\def\bR {{\mathbb{R}}}

\def\cI {{\mathcal I}}

\def\cJ {{\mathcal J}}

\makeatletter
\@addtoreset{equation}{section}

\makeatother

\newtheorem{theorem}{Theorem}[section]
\newtheorem{lemma}[theorem]{Lemma}

\newtheorem{corollary}[theorem]{Corollary}

\theoremstyle{remark}
\newtheorem{remark}[theorem]{\bf Remark}

\newenvironment{demo}{\noindent {\it Proof.--}
      \begin{quotation}\noindent}{\end{quotation}\hfill$\square $}

\begin{document}

\bibliographystyle{plain}

\title{ Schr\"odinger operators on a half-line with inverse square potentials}

\author {Hynek Kova\v{r}\'{\i}k}

\address {Hynek Kova\v{r}\'{\i}k, DICATAM, Sezione di Matematica, Universit\`a degli studi di Brescia, Italy}

\email {hynek.kovarik@ing.unibs.it}

\author {Fran\c{c}oise Truc}

\address {Fran\c{c}oise Truc, Unit{\'e} mixte de recherche CNRS-UJF 5582,
BP 74, 38402-Saint Martin d'H\`eres Cedex (France) }

\email {francoise.truc@ujf-grenoble.fr}

\begin {abstract}
We consider Schr\"odinger operators $H_\alpha$ given by equation \eqref{ham} below.
We study 
the asymptotic behavior of the spectral density $E(H_\alpha, \lambda)$ for $\lambda \to 0$ and the 
$L^1\to  L^\infty$ dispersive estimates associated to the evolution
 operator $e^{-i t H_\alpha}$. In particular we prove that for positive values of
 $\alpha$, the spectral density $E(H_\alpha, \lambda)$ tends to zero as $\lambda \to 0$  with higher
 speed compared to the spectral density of Schr\"odinger operators with a short-range potential $V$. 
We then show how the long time behavior of  $e^{-i t H_\alpha}$ depends 
on $\alpha$. More precisely we show that the decay rate of  $e^{-i t H_\alpha}$  for $t\to\infty$  
can be made arbitrarily large provided we choose $\alpha$ large enough and consider  a suitable operator norm.
\end{abstract}
\maketitle

\section{\bf Introduction}
\label{sec-intro}
\noindent This paper is concerned with Schr\"odinger operators 
\begin{equation} \label{ham}
H_\ga =-\frac{d^2}{dx^2} + \frac{\alpha}{x^2}\, , \qquad    \alpha \, \geq \, -\frac 14, 
\end{equation}
in $L^2(\bR^+)$ with Dirichlet condition at $x=0$. In particular, we are interested in the dependence of various 
spectral properties of $H_\alpha$ on the parameter $\alpha$. Note that potentials of the type
 $\alpha/x^2$ have a special role, since the resulting operator $H_\alpha$  is
scaling invariant. Moreover, it is known that 
the potentials which satisfy $V(x) \sim x^{-2}$ as $x\to\infty$ represent a borderline case 
for certain important spectral inequalities such as dispersive or Strichartz estimates, 
see \cite{gvv}.

It is therefore not surprising that Schr\"odinger operators with inverse square
 potentials have recently attracted certain attention; we 
might mention for example the heat kernel bounds obtained in \cite{ms,ms2}, or Strichartz estimates 
in dimension three studied in \cite{bpst1,bpst2}. Inverse 
square potentials appear naturally also in connection with two-dimensional Schr\"odinger
 operators with Aharonov-Bohm-type magnetic field, see 
\cite{fffp, gk}.  

\smallskip

\noindent Main objects of our interest here are the spectral density
\begin{equation} \label{stone}
E(H_\alpha, \lambda) = \frac{1}{2\pi i}\ \lim_{\eps\to 0+} \left( (H_\alpha-\lambda-i\eps)^{-1} - (H_\alpha-\lambda+i \eps)^{-1} \right), \qquad \lambda >0
\end{equation}
of $H_\alpha$, and the unitary group $e^{-i t H_\alpha}$. In particular, we are going to study 
the asymptotic behavior of $E(H_\alpha, \lambda)$ for $\lambda \to 0$ and the 
$L^1\to  L^\infty$ dispersive estimates associated to the evolution
 operator $e^{-i t H_\alpha}$. It is very well-known that the asymptotic behavior of 
$E(H_\alpha, \lambda)$ for small $\lambda$ is closely related to the asymptotic behavior
 of $e^{-i t H_\alpha}$ for large $t$. There is a huge amount of literature on this subject, 
see e.g. \cite{eg, gsch, JK, mu, sch, sch2, wed, wed2} and references therein. We are not going to 
discuss this connection any further since it will not be used in our proofs. 

For general one-dimensional Schr\"odinger operators of the type $H_V = -\frac{d^2}{dx^2} + V$ 
the behavior of both $E(H_V, \lambda)$ and 
$e^{-i t H_V}$ is known provided the potential $V$ decays fast enough at infinity. In particular,
if zero is a regular point of $H_V$, (which is the generic case), then 
\begin{equation} \label{E-short}
E(H_V, \lambda) \ \sim \  \lambda^{\frac 12} , \qquad \lambda \to 0,
\end{equation}
in a suitable operator topology, see \cite{go, mu, sch2, wed2}. Accordingly, for such short range potentials, under certain regularity conditions, Murata \cite{mu} proved
\begin{equation} \label{murata}
\|\, w^{-1}\, e^{-i t H_V}\, w^{-1} \, \|_{L^2(\bR) \to L^2(\bR)} \ \leq \ C \,  t^{-\frac 32} \qquad \forall \ t>2,
\end{equation} 
where $w$ is a weight function with a sufficient growth at infinity. The corresponding  $L^1\to  L^\infty$ was established by Schlag 
\begin{equation} \label{schlag}
\|\, \rho^{-1}\, e^{-i t H_V}\, \rho^{-1} \, \|_{L^1(\bR) \to L^\infty(\bR)} \ \leq \ C \,  t^{-\frac 32} \qquad \forall \ t>2,
\end{equation} 
with $\rho(x) =(1+ |x|)$, see \cite{sch2}. It is important to mention that the decay
 conditions on $V$, under which all the above results were obtained, imply that $V(x) = o(x^{-2})$ 
as $|x|\to\infty$. 

\smallskip

The goal of the present note is to show that if $V$ is of type $\alpha\, x^{-2}$ with $\alpha >0$,
 then the asymptotic relation \eqref{E-short} is no longer valid and has to replaced by a new one, 
 and, on the other hand, the estimates \eqref{murata} and \eqref{schlag} can be improved. 
 In particular $E(H_\alpha, \lambda)$ decays {\em faster} to zero than in \eqref{E-short}, 
see Theorem \ref{thm-1}. Accordingly the decay in the dispersive estimate \eqref{schlag}
 can be improved provided the weight function $\rho$ grows fast enough at infinity, see 
Theorem \ref{thm-2}. Although our results regard a family of Schr\"odinger operators with explicit
potentials, it can be expected that similar results should hold also if $H_\alpha$ is perturbed by 
a sufficiently short-range perturbation.

It should be finally mentioned that our main results, i.e. Theorems  \ref{thm-1} and \ref{thm-2}, fail in the case of 
Schr\"odinger operators on the whole line due to the presence of the zero resonance.

\section{\bf Main results}
\label{sec-results}

\subsection{Notation} We set $\rho(x) =1+x$ on $\bR^+$. For any $s\in\bR$ we denote
$$
L^2_s(\bR^+) = \{ u : \|  \rho^{\, s}\,  u\|_{L^2(\bR^+)} < \infty \}, \qquad \|u\|_{0,s} :=  \|  \rho^{\, s}\,  u\|_{L^2(\bR^+)}.
$$
Let $\B(s,s')$ be the space of bounded linear operators from $L^2_s(\bR^+)$ to $L^2_{s'}(\bR^+)$ and let $\|\cdot\|_{\B(s,s')}$ denote the 
corresponding operator norm. Finally, we put 
\begin{equation} \label{nu}
\nu =\sqrt{1/4 +\ga}.
\end{equation}

\noindent We have

\begin{theorem} \label{thm-1}
Let $\alpha > -1/4$. Then for any $\eps>0$ and any $s \geq \nu +1 +\eps$ it holds 
\begin{equation} \label{eq-density}
E(H_\alpha, \lambda) = E_0\, \lambda^\nu + \mathcal{O}(\lambda^{\nu+\eps}) \qquad \lambda\to 0+
\end{equation}
in $\B(s,-s)$, where $E_0$ is the integral operator in $L^2(\bR^+)$ with the kernel 
$$
E_0(x,y) = \frac{ (x\, y)^{\nu+\frac 12}}{2^\nu\, \Gamma^2(\nu+1)}.
$$
\end{theorem}

\begin{remark}
 Equation \eqref{eq-density} shows  that for positive values of $\alpha$ the 
density $E(H_\alpha, \lambda)$ is of lesser order than in the case of a short-range potential, see equation \eqref{E-short}.
\end{remark}

\begin{remark}
For a throughout discussion of threshold expansion of resolvents of one-dimensional operators with short-range potentials we refer to \cite{jn}. Asymptotic behaviour of Schr\"odinger groups generated by operators with inverse square decay on conical manifolds was studied in \cite{wa} in the setting of weighted $L^2-$spaces. 
\end{remark}

\begin{theorem} \label{thm-2}
Let $\alpha \geq -1/4$. Then for any $s\in [0, \nu+1/2]$ there exists a constant $C(\alpha,s)$ such that  
\begin{equation} \label{eq-t-decay}
\|\, \rho^{ -s}\, e^{-i t H_\alpha}\, \rho^{-s} \, \|_{L^1(\bR^+) \to L^\infty(\bR^+)} \ \leq \ C(\alpha,s)\ t^{-\frac 12-s} \qquad \forall \ t>0.
\end{equation}
\end{theorem}

\smallskip

\begin{remark}
For $-1/4 < \alpha \leq 0$ the dispersive estimate \eqref{eq-t-decay} can be derived from \cite[Thm.1.11]{fffp} by considering the restriction of inequality 
\cite[Eq.(1.29)]{fffp} to radial functions. On the other hand, the result for $\alpha>0$, namely the faster decay of $ e^{-i t H_\alpha}$ in $t$ is new. The 
maximal decay rate  $t^{-1-\nu}$, achieved by the choice $s=\nu+1/2$,  should be compared with the $t^{-\frac 32}$ decay rate in the estimate \eqref{schlag}. 
\end{remark}

\begin{remark}
Note also that in the border-line case $\alpha=-1/4$, which means $\nu=0$, Theorem \ref{thm-2} with the choice $s=1/2$ gives the decay rate $t^{-1}$, which is  the decay rate 
of the free evolution $e^{i t \Delta}$ in dimension two. This is not surprising since the operator $-\frac{d^2}{dx^2} - \frac{1}{4 x^2}$ in $L^2(\bR^+)$ 
with Dirichlet boundary condition at zero is unitarily equivalent, by means of the unitary mapping $f(x)\mapsto \sqrt{x} \ f(x)$, to the Laplacian $-\Delta$ 
in $L^2(\bR^2)$ restricted to radial functions.  
\end{remark}

\noindent An immediate consequence of Theorem \ref{thm-2} is the following 

\begin{corollary}
Let $\alpha \geq -1/4$. Then for any $s\in [0, \nu+1/2]$ and any $\beta >s+1/2$ there exists a constant $C_2$, depending only on $\alpha, \beta$ and $s$, such that  
\begin{equation} \label{L2}
\|\, \rho^{-\beta}\, e^{-i t H_\alpha}\, \rho^{-\beta} \, \|_{L^2(\bR^+) \to L^2(\bR^+)} \ \leq \ C_2\ t^{-\frac 12-s} \qquad \forall \ t>0.
\end{equation}
\end{corollary}

\begin{proof} 
Let $u\in L^2_\beta(\bR^+)$ and let $f = \rho^{s}\, u$. Then by the Cauchy-Schwarz inequality $f\in L^1(\bR^+)$ and 
\begin{align} \label{c-s}
\| f \|^2_{L^1(\bR^+)} & =  \left( \int_0^\infty \rho^{s}\, \rho^{-\beta}\, \rho^\beta\, u\, dx \right)^2 \leq\  C_1\, \|\, \rho^\beta\, u\|^2_{L^2(\bR^+)},
\end{align}
where we have used the fact that $\beta> s+1/2$. Hence from Theorem \ref{thm-2} and \eqref{c-s} we obtain
\begin{align*}
\|\, \rho^{-\beta}\, e^{-i t H_\alpha}\, u\, \|^2_{L^2(\bR^+)} & = \int_0^\infty \rho(x)^{2s-2\beta}\, | \, \rho^{-s}\, e^{-i t H_\alpha}\, \rho^{-s}\, f\, |^2\, dx  \\
& \leq \, C^2(\alpha,s)\ t^{-1-2s}\, \| \, f\|^2_{L^1(\bR^+)} \\
& \leq \, C^2(\alpha,s)\, C_1 \ t^{-1-2s}\ \|\, \rho^\beta\, u\|^2_{L^2(\bR^+)}.
\end{align*}
This proves \eqref{L2}.
\end{proof}

\noindent Inequality \eqref{L2} should be compared with the estimate \eqref{murata} valid for short-range potentials.

\section{\bf Proofs}
\subsection{Proof of Theorem \ref{thm-1}} 
For simplicity we shall drop the index 
$\alpha$ in the sequel and write $E(\gl)$ instead of $E(\ga ,\gl)$. We also use the notation 
$$
R(\gl,x,y) = \lim_{\eps\to 0+} (H_\alpha-\lambda -i \eps)^{-1} (x,y).
$$
We first study the solutions $u \in L^2 (\bR^+)$ of the generalized eigenvalue equation 
\begin{equation}\label{bessel}
-u''-\frac{\alpha}{x^2}= \gl u.
\end{equation}
After setting $u(x)= \sqrt x \, \psi (\sqrt{\gl}\, x)$, equation \eqref{bessel} writes
\begin{equation}\label{bess}
z^2\psi''-x +z\, \psi' + (z^2-\nu^2)\, \psi= 0\ ,
\end{equation}
with $z=\sqrt{\gl} \, x$. The latter is a Bessel equation of the first kind, see \cite[Sec.9.1]{as}. We now find two solutions $u_1, \, u_2$ of \eqref{bessel} 
which satisfy $u_1(0)=0$ and $u_2 \in L^2(1,\infty)$ for Im$\, \lambda >0$.  Since 
$$
|\, J_\nu (z) +i \, Y_\nu (z)  \, | \, \sim \, \sqrt{\frac{2}{z\pi}}\  |\, e^{\, i z}\, | \, ,  \qquad |z| \to \infty, \quad \text{Im} \, z >0,
$$
see \cite[Eqs.9.1.3, 9.2.3]{as}, and 
\begin{equation} \label{j-zero}
J_\nu(z) = \frac{(z/2)^\nu}{\Gamma(\nu+1)} + o(z^\nu) \qquad z\to 0. 
\end{equation}
by \cite[Eq.9.1.7]{as}, the sought solutions $u_1$ and $u_2$ take the form 
\begin{align}
u_1(x) &=  \sqrt x \ J_\nu (\sqrt{\gl}\, x)  \label{sol1} \\
u_2(x) &= \sqrt x\  (J_\nu (\sqrt{\gl}\, x) +i \, Y_\nu (\sqrt{\gl}\, x)) \label{sol2} \ .
\end{align}
Hence by the theory of Sturm-Liouville problems we obtain the resolvent kernel 
 \begin{align}
  R(\gl,x,y) & = \frac{i\pi}{2} \sqrt {xy} \ J_\nu (\sqrt{\gl} \,x)\, (J_\nu (\sqrt{\gl}\, y)
 +i\, Y_\nu (\sqrt{\gl} \, y))
\qquad (x\leq y )  \label{res} \\
  R(\gl,x,y) & = \frac{i\pi}{2} \sqrt {xy}\  J_\nu (\sqrt{\gl}\, y)\, (J_\nu (\sqrt{\gl}\, x)
 +i\, Y_\nu (\sqrt{\gl}\, x))
\qquad (x\geq y )\label{reso}
\end{align}
The Stone formula \eqref{stone} then implies that
\begin{equation}\label{dens}
  E(\gl,x,y) = \frac{1}{\pi}\, {\rm Im}\, R(\gl,x,y)  = \frac{1}{2} \sqrt {xy}\  J_\nu (\sqrt{\gl}\, x)\, J_\nu (\sqrt{\gl}\, y)\ .
\end{equation}
From \eqref{j-zero} we now easily verify that 
\begin{equation} \label{E0}
\lim_{\lambda\to 0+}\, \lambda^{-\nu}\,  E(\gl,x,y) =  \frac{ (x\, y)^{\nu+\frac 12}}{2^\nu\, \Gamma^2(\nu+1)} = E_0(x,y). 
\end{equation}
Let us define the rest term $E_1(\lambda)$ as the integral operator in $L^2(\bR^+)$ with the kernel given by 
\begin{equation}\label{densi}
  E_1(\gl,x,y) =  E(\lambda, x,y)  -   E_0(x,y) \, \gl^\nu. 
\end{equation}

\smallskip

\noindent To prove  Theorem \ref{thm-1} we need the following

\begin{lemma}\label{first}
For any $\eps>0$ and any $s > \nu +1+\eps$ we have 
\begin{equation} \label{rest}
\| \,  E_1(\gl)\,  \|_{\B(s,-s)}  = \mathcal{O}(\gl^{\nu+\eps}) \qquad \lambda\to 0+.
\end{equation}
\end{lemma}

\begin{demo} We will use the fact that 
\begin{equation} \label{equiv}
\| \,  E_1(\gl)  \, \|_{\B(s,-s)}  = \| \rho^{-s} E_1(\gl) \rho^{-s} \|_{L^2(\bR^+)\to L^2(\bR^+)} .
\end{equation}
From (\ref {densi}) we get that
\begin{equation}\label{Tayl}
 \rho^{-s} \gl^{-\nu} E(\gl)\,  \rho^{-s}=  \rho^{-s}\, E_0\, \rho^{-s} + \rho^{-s} \gl^{-\nu}
 E_1(\gl)\,  \rho^{-s} .
\end{equation}
Note that the operator $E_0$ is Hilbert-Schmidt in $\B(s,-s)$. This follows from the identity \eqref{equiv} applied to $E_0$. 
Hence by applying the Taylor formula to the operator $\rho^{-s} \gl^{-\nu} E(\gl)\,  \rho^{-s}$ at $\gl =0$ we find that   
the claim of the Lemma will follow if we show that 
\begin{align}
 \| \, \rho^{-s} \, \partial_\gl(\gl^{-\nu} E(\gl))\, \rho^{-s} \|_{HS(\bR^+)} & = \| \, \rho^{-s} \, \partial_\gl(\gl^{-\nu} E_1(\gl))\, \rho^{-s} \|_{HS(\bR^+)}  \nonumber \\
  & =  \mathcal{O} (\gl^{-1+\eps})\  \quad \lambda\to 0, \label{bdnorm}
\end{align}
where $\|\cdot \|_{HS(\bR^+)}$ denotes the Hilbert-Schmidt norm in $L^2(\bR^+)$. 
Using the recurrence relations for the derivatives of $J_\nu$: 
\begin{align*}
J'_\nu (z) & = -J_{\nu +1} (z) + \frac{\nu}{z}\,  J_\nu (z) \\
J'_\nu (z) & = J_{\nu +1} (z) - \frac{\nu}{z}\,  J_\nu (z),
\end{align*}
see \cite[Eq.9.1.27]{as},  we get from (\ref{dens})
\begin{align}
 \partial_\gl (\gl^{-\nu}E(\gl,x,y)) &=
 -\frac{\gl^{-\nu-1/2}}{4 }\sqrt {x\, y}\ \Big[ \ (x\,  J_{\nu +1} 
(\sqrt{\gl}\, x)\,  J_\nu (\sqrt{\gl}\,  y)\nonumber \\
& \qquad \qquad\qquad \qquad \ 
+ y\,  J_{\nu +1} (\sqrt{\gl}\, y)\, 
J_\nu (\sqrt{\gl} \, x)\ \Big] \ . \label{deriv}
\end{align}
Hence by the Cauchy-Schwarz inequality 
\begin{align}
& \|  \rho^{-s} \, \partial_\gl(\gl^{-\nu} E(\gl))\, \rho^{-s} \|_{HS(\bR^+)} ^2  = \nonumber \\
 & \qquad  \qquad= 
 \int_0^{\infty} \int_0^{\infty} |\partial_\gl (\gl^{-\nu}E(\gl,x,y))|^2 \, \rho(x)^{-2s} \rho(y)^{-2s}\, dxdy \nonumber\\
 &  \qquad \qquad \leq \ C\, \lambda^{-1-2\nu}\  \cI(\lambda)\, \cJ (\lambda)  \label{int}
 \end{align}
where
\begin{align*}
\cI(\lambda) & = \int_0^{\infty} x^3\,  J_{\nu +1}^2 (\sqrt{\gl}\, x)\,  (1+x)^{-2s} dx \\
\cJ(\lambda) & = \int_0^{\infty} y \, J_{\nu}^2 (\sqrt{\gl}\, y)\,  (1+y)^{-2s} dy
\end{align*}
and $C$ is a constant independent of $\lambda$. To estimate the last two integrals we will need a point-wise 
estimate on the Bessel function $J_\nu$. From the integral representation 
$$
J_\nu(z) = \frac{2 \, (\frac z2)^\nu}{\sqrt{\pi}\ \Gamma(\nu+\frac 12)}\, \int_0^1 (1-t^2)^{\nu-\frac 12}\, \cos(zt)\, dt,
$$
see \cite[Eq.9.1.20]{as}, it follows that $| J_\nu(z)| \leq C_\nu\, z^\nu$ for all $z>0$ and $\nu>0$. 
On the other hand, by \cite[Eq.9.1.20]{as} we have $| J_\nu(z)| \leq 1$ for all $z>0$ and $\nu\geq 0$. A combination of these two upper bounds then 
implies that 
for any $-1/2 \leq \mu\leq \nu$ there exists a constant $C(\mu,\nu)$ such that 
\begin{equation} \label{j-upperb}
| J_\nu(z)| \ \leq \ C(\mu,\nu)\ z^\mu \qquad \forall\ z>0, \quad \forall\  \mu \in \left[-\frac 12, \, \nu\right ]. 
\end{equation}
Using \eqref{j-upperb} with $\mu=\nu-1+2\eps$ in $\cI(\lambda)$ and with $\mu=\nu$ in $\cJ(\lambda)$ together with the fact that $s \geq \nu+1+\eps$, we find 
$$
\cI(\lambda) = \mathcal{O}(\gl^{\nu-1+2\eps}) ,  \qquad \cJ(\lambda) = \mathcal{O}(\gl^{\nu})  \qquad \lambda\to 0.
$$
In view of  \eqref{int} this implies \eqref{bdnorm} and therefore completes the proof. 
\end{demo}

\noindent Theorem \ref{thm-1} now follows from \eqref{densi} and \eqref{rest}.

\subsection{Proof of Theorem \ref{thm-2}}
 We will prove Theorem \ref{thm-2} by estimating the integral kernel  of the operator 
$e^{-it H_\alpha}$. To provide a formula for the integral kernel, we will follow  \cite[Sec.5]{k11}, where the formula for the integral of the heat semi-group $e^{-t H_\alpha}$ was established, see also \cite{gk}. Equation \eqref{dens} in combination with the  Weyl-Titchmarsh-Kodaira Theorem, cf. \cite[Chap.13]{ds}, shows that the operator $H_\alpha$ is unitarily equivalent to a multiplication operator, namely we have
\begin{equation} \label{diag}
\U_\nu \, H_\alpha \, \U_\nu ^{-1}\, f(p) = p\, f(p), \qquad f \in \U_\nu (D(H_\alpha)),
\end{equation}
where $D(H_\alpha)$ denotes the operator domain of $H_\alpha$ and the mappings 
$\U_\nu ,\, \U_\nu^{-1}: L^2(\bR_+) \to L^2(\bR_+)$ are given by
\begin{equation}\label{fb}
\begin{aligned} 
(\U_\nu \, g)(p) & = \int_0^\infty g(x) \sqrt{x}\, J_\nu(x\sqrt{p})\, dx \\
(\U_\nu ^{-1} f)(x) & = \frac 12\, \int_0^\infty f(p)\sqrt{x}\, J_\nu(x\sqrt{p})\, dp
\end{aligned} 
\end{equation}
The mapping $\U_\nu$ and $\U_\nu^{-1}$ define unitary operators on $L^2(\bR_+)$. Let $g\in C_0^\infty(\bR^+)$. By \cite[Thm.3.1]{T}
\begin{equation} \label{limit}
e^{-i t H_\alpha }\, g  = \lim_{\eps\to 0+} e^{-(\eps+i t)\, H_\alpha }\, g.
\end{equation}
In view of \eqref{diag} we thus get
\begin{align}
 & \lim_{\eps\to 0+} \big( e^{-(\eps+i t) H_\alpha }\, g\big)(r)  = 
 \lim_{\eps\to 0+}  \big( \U_\nu ^{-1}\, e^{-(\eps+i t)
\, p}\, \U_\nu \, g\big)(x) \nonumber \\
& \qquad \qquad = \lim_{\eps\to 0+}  \frac 12\, \int_0^\infty \sqrt{xy}
\int_0^\infty e^{-(\eps+i t)\, p}  J_\nu (x\sqrt{p})
J_\nu(y \sqrt{p})\,  dp\,  g(y) \, dy \nonumber \\
& \qquad \qquad= \lim_{\eps\to 0+} \frac{1}{2 (\eps+i t) }\, \int_0^\infty \sqrt{xy}\, \,
I_\nu \left(\frac{x y}{2(\eps+i t)}\right)\,
e^{-\frac{x^2+y^2}{4(\eps+i t)}}\, g(y)\, dy.  \label{epsilon}
\end{align}
where we have used \cite[Eq.4.14(39)]{erde} to calculate the integral with respect to
$p$. Moreover, from  \cite[Eq.9.6.18]{as} it follows that the function 
\begin{equation} \label{I-estim}
 I_\nu \Big(\frac{xy}{2(\eps+i t)}\Big ) \, e^{-\frac{x^2+y^2}{4(\eps+i t)}} 
\end{equation}
is bounded on every compact interval uniformly with respect to $\eps>0$.
 Since the support of $g$ is compact, we can use the dominated theorem and interchange the 
limit and integration in \eqref{epsilon}. Taking the limit $\eps\to 0$ and using the 
identity $I_\nu( iz) = e^{-i \nu\pi/2} J_\nu(z)$, see \cite[Eq.9.6.3]{as}, we obtain
\begin{equation} \label{eq-kenrel}
 \big( e^{-i t H_\alpha }\, g\big) (x)  = \frac{1}{2 i t }\, \int_0^\infty \sqrt{xy}\, \,
J_\nu \left(\frac{x y}{2 t}\right)\, e^{-\frac{x^2+y^2}{4 i t}}\, 
 e^{-\frac{i \nu\pi}{2}} \, g(y)\, dy.
\end{equation}
Now we apply the upper bound \eqref{j-upperb} with $\mu =s-1/2 \in [-1/2, \nu] $ and $z= \frac{x y}{2 t}$.
 This yields
\begin{equation} \label{eq-sup}
\sup_{x,y \in\bR^+}\, \Big |\,  \rho(x)^{-s}\, \sqrt{xy}\, \, J_\nu \left(\frac{x y}{2 t}\right)\, \rho(y)^{-s} \Big| \ < \, C(\alpha,s)\ t^{\frac 12-s}. 
\end{equation} 
The last equations now imply that 
$$
\|\, \rho^{ -s}\, e^{-i t H_\alpha}\, \rho^{-s}\, f \, \|_{L^\infty(\bR^+)} \ \leq \ C(\alpha,s)\ t^{-\frac 12-s}\, \| f\|_{L^1(\bR^+)}
$$
for all $f\in L^1(\bR^+)$. This proves inequality \eqref{eq-t-decay}. 


\section*{\bf Acknowledgements}
H.K. would like to thank the Institut Fourier in Grenoble  for the warm hospitality extended to him during his stay. 
The work of H.K. has been partially supported by the MIUR-PRIN'2010-11 grant for the project  ''Calcolo della variation''. Both authors thank the referee whose remarks and comments helped them improve the original version of the text. 

\bibliographystyle{amsalpha}

\end{document}